\documentclass[10pt]{article} 
\usepackage{amsmath,amsfonts,latexsym,enumitem}
\usepackage[noend]{algpseudocode}
\usepackage{xcolor,algorithm,multirow,lscape,graphicx}	
\usepackage{tikz}
\usepackage{pgfgantt}
\usepackage{hyperref}
\usepackage[capitalise,noabbrev]{cleveref}
\normalsize
\newtheorem{theorem}{Theorem}
\newtheorem{lemma}{Lemma}
\newtheorem{corollary}{Corollary}
\newtheorem{proposition}{Proposition}

\newtheorem{definition}{Definition}

\newenvironment{proof}{{\sc Proof. }}{\hfill$\Box$\vspace{0.1in}}
\newcommand{\OPT}{{\rm OPT}}
\newcommand{\pw}{{\rm pw}}
\newcommand{\tw}{{\rm tw}}

\newcommand{\Qge}{\mathbb{Q}_{\ge 0}}
\newcommand{\mcA}{\mathcal{A}}
\newcommand{\mcB}{\mathcal{B}}

\title{Layering and reverse extension for weighted $k$-path vertex cover}
\author{
Jiajing~Wu\thanks{School of Mathematics, Hangzhou Dianzi University. Hangzhou, China.
	Emails: \texttt{\{265070205,anzhang,90077\}@hdu.edu.cn}}
\and
An~Zhang$^*$\thanks{Corresponding authors.}
\and
Shiming~Li$^*$
\and
Guohui~Lin\thanks{Department of Computing Science, University of Alberta.  Edmonton, Canada.
	Email: \texttt{guohui@ualberta.ca}} $^{\dagger}$
}
\date{}
\begin{document}
\maketitle
\begin{abstract}
We study the weighted $k$-path vertex cover problem for any fixed integer $k \ge 2$,
in which a vertex non-negatively weighted simple graph $G = (V, E)$ is given and one wants to find a minimum-weight vertex subset $C \subseteq V$
such that every simple path on $k$ vertices in $G$ contains at least one vertex of $C$.
The problem generalizes the classic weighted vertex cover problem, that is, $k = 2$, and it admits a $k$-approximation algorithm of Hochbaum~[SIAM J. Comput., 1982]
and a $(k - \frac 12)$-approximation algorithm of Br\"ustle et al.~[IWOCA 2021].
In this paper, we carefully design local weight functions in the Layering and Reverse Extension local ratio framework of Bar-Yehuda et al.~[ACM Comput. Surv., 2004],
leading to an approximation algorithm of ratio $\frac {3k-1}4$ for odd $k \ge 3$ and of ratio $\frac {3k}4$ for even $k \ge 4$.
When $k = 3$ and $k = 4$, these ratios are $2$ and $3$, respectively,
matching the state-of-the-art approximation ratios for the corresponding special cases.

\end{abstract}

\subsection*{Acknowledgement}
This work is supported by the NNSF of China (Grants 12371316, 92567201),
the PNSF of Zhejiang China (Grant LZ25A010001), and
the NSERC Canada.

\newpage
\section{Introduction}\label{sec:intro}
Throughout the paper, we consider a vertex non-negatively weighted simple graph $G = (V, E)$.
In the special case where the vertices have the uniform weight, the graph is also called {\em unweighted}.
Unless otherwise explicitly specified, the graphs and the problems are vertex-weighted.

For an integer $k \ge 2$, a {\em $k$-path} in the graph $G$ is a simple path on $k$ vertices.
A vertex subset $C \subseteq V$ is a {\em $k$-path vertex cover}, denoted as $k$PVC, if every $k$-path in $G$ contains at least a vertex of $C$,
or equivalently, the subgraph of $G$ obtained by removing the vertices of $C$ and those edges incident at the vertices of $C$, i.e., $G[V \setminus C]$, does not contain any $k$-path.
In the latter sense, $C$ is also called a {\em deletion set} for convenience.
One sees that when $k = 2$, a $2$-path is an edge, and a $2$PVC is a classic vertex cover.

The weighted $k$PVC problem is to find a $k$PVC $C$ to minimize the total weight of the vertices in $C$, i.e., $w(C) = \sum_{v \in C} w(v)$,
where $w(\cdot)$ is the given nonnegative vertex weight function for the input graph.
The problem on unweighted graphs was first introduced by Bre\v{s}ar et al.~\cite{BKK11},
and later the weighted version was first studied by Bre\v{s}ar et al.~\cite{BKS14},
which generalize the classic vertex cover problem and its weighted variant, respectively.
In this paper, we study the weighted $k$PVC problem for any fixed integer $k \ge 3$, from the perspective of approximation algorithms.

\subsection{Previous work}
The classic vertex cover, or $2$PVC, problem and its weighted variant are NP-hard~\cite{GJ79} and APX-hard~\cite{AK00}, even on cubic graphs.
They admit a $2$-approximation algorithm~\cite{BE81} but, assuming the Unique Games Conjecture (UGC), are not approximable within $2 - \epsilon$ for any positive $\epsilon$~\cite{KR08}.
In the pioneering work of Bre\v{s}ar et al.~\cite{BKK11}, unweighted $k$PVC is shown NP-hard, and consequently the weighted variant is NP-hard too, for every $k \ge 3$.

On the approximability, given a vertex-weighted graph $G = (V, E)$,
one can construct an instance of the {\em weighted} set cover problem in which every $k$-path of $G$ becomes an ``{\em element}'' and
every vertex of $V$ becomes a set that ``{\em covers}'' all the $k$-paths passing it and has a weight.
In particular, one observes that in such an instance, every element occurs in exactly $k$ sets (typically the number of sets is called {\em frequency}).
Therefore, Bar-Yehuda and Even's~\cite{BE81} and Hochbaum's~\cite{Hoc82} algorithms for the low-frequency weighted set cover problem can be employed, 
to become the first $k$-approximation algorithm for the weighted $k$PVC problem.
Since its first formal study~\cite{BKS14} in 2014, the only improvement was made in 2021 by Br\"ustle et al.~\cite{BEH21},
who presented a $(k - \tfrac 12)$-approximation algorithm.

For the unweighted $k$PVC problem, Lee presents an $O(\log k)$-approximation algorithm~\cite{Lee19}.

On the inapproximability side,
Bre\v{s}ar et al. showed that an $\alpha$-approximation algorithm for the unweighted $k$PVC problem
yields an $\alpha$-approximation algorithm for the classic vertex cover problem~\cite{BKK11}.
Therefore, assuming the UGC, the unweighted, and consequently the weighted,
$k$PVC problem is not approximable within $2 - \epsilon$ for any positive $\epsilon$.

Besides the above general results,
there are also studies on some specific low values of $k$ and on some restricted graph classes.
For instance, Tu and Zhou gave two $2$-approximation algorithms for weighted $3$PVC, based on the primal-dual method and the {\em layering} method~\cite{TZ11a,TZ11b}, respectively.
For weighted $4$PVC, Camby et al. gave a $3$-approximation algorithm~\cite{CCC14}.
These are the state-of-the-art results for $k = 3$ and $4$, respectively.

When the graphs are unweighted and regular, better constant approximation ratios are known for $k \ge 3$~\cite{RSU15,ZCC20}.
For unweighted $3$PVC, a {\em polynomial-time approximation scheme} (PTAS) is known for planar graphs~\cite{TS19},
and a PTAS is known for bounded-heterogeneity ball graphs~\cite{ZLS17}.

Exact polynomial-time algorithms are also known for several restricted graph classes.
Bre\v{s}ar et al.~\cite{BKK11} showed that unweighted $k$PVC is solvable in linear time for trees;
Tu~\cite{Tu17} gave an exact polynomial-time algorithm for cactus graphs.
For weighted $k$PVC, Bre\v{s}ar et al.~\cite{BKS14} gave exact polynomial-time algorithms for complete graphs, paths, cycles, and trees;
in particular, their algorithm for trees runs in $O(k |V|)$ time.
When $k = 3$, since unweighted $3$PVC is dual to the {\em maximum dissociation set} problem,
the results of Orlovich et al.~\cite{ODF11} imply the polynomial-time solvability of unweighted $3$PVC on chordal graphs,
weakly chordal graphs, interval-filament graphs, and AT-free graphs.
Bre\v{s}ar et al.~\cite{BKK19} also gave a linear-time algorithm for unweighted $3$PVC on $4$-path-tidy graphs%
\footnote{A partner of an induced $4$-path in a graph is a vertex outside the path such that the induced subgraph on these five vertices contains at least two $4$-paths.
A graph is $4$-path-tidy if every induced $4$-path has at most one partner~\cite{GRT97}.}
Further exact, approximation, and parameterized algorithmic results have been summarized in a survey~\cite{Tu22}.

\subsection{Our contribution}
We study the weighted $k$PVC problem for any fixed integer $k \ge 3$, from the approximation algorithm perspective.
We adopt the Layering and Reverse Extension local ratio framework of Bar-Yehuda et al.~\cite{BBF04},
and carefully design local cost functions, leading to an approximation algorithm of ratio $\frac {3k-1}4$ for odd $k \ge 3$ and of ratio $\frac {3k}4$ for even $k \ge 4$.
When $k = 3$ and $k = 4$, these ratios are $2$ and $3$, respectively,
matching the state-of-the-art approximation ratios for the corresponding special cases.

\subsubsection{The algorithm framework}
The Layering and Reverse Extension framework has a forward phase of recursive layering and a backward phase of solution extension,
with its performance analyzed by the local ratio technique~\cite{BBF04}.
During the forward layering phase,
we abuse $G = (V, E)$ to denote the current remaining graph in which every vertex of $V$ has a positive remaining weight.
In the current iteration, a vertex subset $U \subseteq V$ is located such that the induced subgraph on $U$, $G[U]$,
contains a spanning subgraph isomorphic to a carefully designed structure $H$, which in our case is a maximal-arm spider (in \Cref{def02}).
Then, with respect to a local cost function designed for the vertices of $U$,
a scaling factor is found and the vertices of $U$ of the scaled weights form a layer, which is subtracted from the graph $G$ to give rise to at least one zero-weighted vertex.
The zero-weighted vertices are then removed, and layering continues into the next iteration.

When an iteration is reached in which the remaining graph $G$ can be solved exactly, which in our case is spider-free,
a {\em minimal} solution is obtained and the second reverse extension phase starts to restore the removed vertices in the reverse order of layering iterations,
and at the end returns a $k$PVC for the original graph.

For the performance analysis, the designed local cost function for layering is used to bound the weight of the achieved solution in each layer
with respect to the optimum local to that layer, following the good-local-cost formulation of Aprile et al.~\cite{ADF23}.

\subsubsection{The main result and paper organization}
\begin{theorem}
\label{thm01}
For every fixed integer $k \ge 3$, the weighted {\rm $k$-Path Vertex Cover} problem admits a polynomial-time $\rho_k$-approximation algorithm, where
\[
\rho_k = \begin{cases}
	\frac {3k - 1}4, & k \mbox{ is odd};\\
	\frac {3k}4,     & k \mbox{ is even}.
	\end{cases}
\]
\end{theorem}

We introduce some key notations and the spider structures in Section~\ref{sec01}.
Section~\ref{sec02} presents a polynomial-time dynamic programming exact algorithm for spider-free graphs, that is,
the input graph does not contain the defined spider structure as a subgraph.
The main algorithm of Layering and Reverse Extension, together with the local-ratio analysis and the running time analysis, are in Section~\ref{sec03}.
Section~\ref{sec04} concludes the paper.

\section{Preliminaries}\label{sec01}
All graphs are finite, simple, and undirected.
Given a graph $G$ and its nonnegative vertex weight function $w: V(G) \to \Qge$ and $C \subseteq V(G)$, let $w(C) = \sum_{v \in C} w(v)$ for any vertex subset $C$, and let
\[
\OPT(G, w) = \min\{w(C): C \mbox{ is a $k$PVC of } G\}.
\]

\begin{definition}[Three-arm spider]
\label{def01}
For three positive integers $a$, $b$, and $c$, let $S(a, b, c)$ be the tree obtained from three paths of lengths $a, b, c$, respectively, by identifying one endpoint of each path.
The identified vertex is the {\em center};
the three paths are the {\em arms} and the vertices other than the center are arm-vertices.

Define the three-arm spider
\begin{equation}
\label{eq01}
  T_k =
  \begin{cases}
    S(m, m, m),     & k = 2m + 1;\\
    S(m - 1, m, m), & k = 2m.
  \end{cases}
\end{equation}
\end{definition}

The number of vertices of $T_k$ is
\begin{equation}
\label{eq02}
  |V(T_k)| =
  \begin{cases}
    1 + 3m = \dfrac {3k - 1}2,          & k \mbox{ odd;}\\
    1 + (m - 1) + m + m = \dfrac {3k}2, & k \mbox{ even.}
  \end{cases}
\end{equation}
If a graph contains no $T_k$ as a subgraph, then we call it {\em $T_k$-free} or {\em spider-free}.

Given a three-arm spider, an $(m+1)$-path (i.e., length-$m$ path) ending at $c$ and containing no arm-vertex of the spider
can be added to become a new arm to make it into a four-arm spider;
a four-arm spider can continue to grow into a five-arm spider, and so on.

\begin{definition}[Maximal-arm spider]
\label{def02}
Starting with a three-arm spider $T_k$ re-written as $\mcA = \{c, Q_1, Q_2, Q_3\}$, where $c$ is the center and $Q_1, Q_2, Q_3$ are the three arms,
repeatedly find an $(m+1)$-path ending at $c$ and containing no arm-vertex of $\mcA$, and then append it as a new arm to $\mcA$.
Upon no new arm can be found, the resultant $\mcA = \{c, Q_1, Q_2, \dots, Q_p\}$, where $p \ge 3$, is called a {\em maximal-arm} spider.
See for an illustration in \Cref{fig01}.
\end{definition}
We remark that a maximal-arm spider is greedily grown from a three-arm spider, and thus it is not necessarily unique,
neither necessarily has a unique number of arms, but all depending on what the new arms are and how they are discovered.

\begin{figure}[ht]
\centering\scalebox{1.0}{
  \setlength{\unitlength}{1bp}%
  \begin{picture}(261.03, 110.13)(0,0)
  \put(0,0){\includegraphics{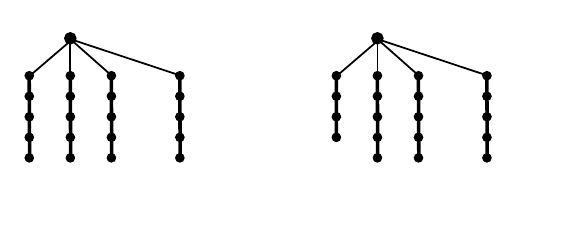}}
  \put(35.23,95.57){\fontsize{11.38}{13.66}\selectfont $c$}
  \put(56.72,53.19){\fontsize{11.38}{13.66}\selectfont $\dots$}
  \put(182.64,95.57){\fontsize{11.38}{13.66}\selectfont $c$}
  \put(204.12,53.19){\fontsize{11.38}{13.66}\selectfont $\dots$}
  \put(5.67,26.58){\fontsize{11.38}{13.66}\selectfont $m$}
  \put(31.95,26.58){\fontsize{11.38}{13.66}\selectfont $m$}
  \put(54.95,26.58){\fontsize{11.38}{13.66}\selectfont $m$}
  \put(87.80,26.58){\fontsize{11.38}{13.66}\selectfont $m$}
  \put(149.38,36.44){\fontsize{11.38}{13.66}\selectfont $m-1$}
  \put(179.35,26.58){\fontsize{11.38}{13.66}\selectfont $m$}
  \put(202.35,26.58){\fontsize{11.38}{13.66}\selectfont $m$}
  \put(235.20,26.58){\fontsize{11.38}{13.66}\selectfont $m$}
  \put(5.67,8.12){\fontsize{11.38}{13.66}\selectfont $k$ odd: $m = \tfrac {k-1}2$}
  \put(157.09,8.12){\fontsize{11.38}{13.66}\selectfont $k$ even: $m = \tfrac k2$}
  \end{picture}%
}
\caption{An illustration of a maximal-arm spider $\mcA = \{c, Q_1, Q_2, \dots, Q_p\}$, where $p \ge 3$, when $k$ is odd (left, where $m = \frac {k-1}2$)
	and when $k$ is even (right, where $m = \frac k2$).\label{fig01}}
\end{figure}

In our Layering and Reverse Extension algorithm for the weighted $k$PVC problem,
we will use maximal-arm spiders to construct layers.
The next two lemmas give the first idea on why they are important to the approximation ratio we want to prove.

\begin{lemma}
\label{lemma01}
For a $p$-arm spider $\mcA = \{c, Q_1, Q_2, \dots, Q_p\}$, if the center $c$ is not included in a $k$PVC $C$, then $C$ contains at least one vertex from each of $p-1$ arms.
\end{lemma}
\begin{proof}
If $C$ does not contain any vertex from two arms, then these two arms together with the center $c$ form a $k$-path, a contradiction.
\end{proof}

\begin{lemma}
\label{lemma02}
If all the vertices of a maximal-arm spider $\mcA$ are positively weighted,
then every $k$PVC $C$ of the graph $G$ containing all the vertices of $\mcA$ can be reduced by removing at least one vertex of $\mcA$.
\end{lemma}
\begin{proof}
Suppose $P$ is a $k$-path in $G$ passing the center $c$.
We claim that $P$ contains at least an arm-vertex of $\mcA$, since otherwise from $P$ one can find a new arm for $\mcA$, contradicting the maximality.
In other words, $C \setminus \{c\}$ is still a $k$PVC.
This proves the lemma.
\end{proof}

Basically, using \Cref{lemma01}, we assign a weight of $1$ to every arm-vertex while a weight of $p-1$ to the center.
This way, the minimum cost of covering a $p$-arm spider is at least $p-1$.
Furthermore, given that all these weights are positive,
by \Cref{lemma02} our to-be computed solution will not contain all the vertices of the spider but at most $pm - 1 + (p-1) = p \cdot \frac {k+1}2 - 2$ when $k$ is odd
(or a similar value when $k$ is even).
It follows that the approximation ratio is no more than $\frac {3k-1}4$.
We will formally present the algorithm and prove its performance in the next two sections.

\Cref{lemma02} also inspires the definition of {\em minimal} $k$PVC:
A $k$PVC $C$ of $G$ is {\em minimal} if for every $v \in C$, $C \setminus \{v\}$ is no longer a $k$PVC.

\begin{lemma}
\label{lemma03}
For a $k$PVC $C$ of $G$, $C$ is minimal if and only if for every $v \in C$, there exists a $k$-path $P$ such that $V(P) \cap C = \{v\}$, which is called the {\em witness path} of $v$.
\end{lemma}
\begin{proof}
Consider a vertex $v \in C$.
Since $C \setminus \{v\}$ is not a $k$PVC, there exists a $k$-path $P$ that is not covered by $C \setminus \{v\}$, i.e., $V(P) \cap (C \setminus \{v\}) = \emptyset$.
On the other hand, since $C$ is a $k$PVC, $P$ must contain at least one vertex of $C$, and thus exactly the vertex $v$.

Conversely, removing $v$ from $C$ will leave $P$ uncovered, and thus $C$ is minimal.
\end{proof}

\section{A polynomial-time exact algorithm for spider-free graphs}\label{sec02}
In this section, we show that the $k$PVC problem on spider-free graphs is solvable in polynomial time.
A spider-free graph is one that does not contain any three-arm spider $T_k$ defined in Eq.~\eqref{eq01},
and thus it is the boundary case terminating the first layering phase of the algorithm.
For distinction, we fix a spider-free graph $H = (V, E)$ for discussion.

There are two steps in the exact algorithm.
First, we show that the spider-free graph $H$ has pathwidth bounded by a function depending only on $k$, and construct a path decomposition of bounded width for $H$.
We then give a dynamic program on this path decomposition to compute a minimum weight $k$PVC for $H$.

\subsection{Bounded pathwidth and a path decomposition}
Given a graph, an {\em edge contraction} merges the two connecting vertices into one and removes the loop just formed.
A graph $F$ is a {\em minor} of $H$ if $F$ can be obtained from $H$ by a sequence of edge contractions, edge and vertex deletions.
An {\em edge subdivision} adds new vertices to the edge whereby breaking the edge into two.
A graph $F$ is a {\em topological minor} of $H$ if a subdivision of $F$ is (isomorphic to) a subgraph of $H$.
One sees that a topological minor is a minor, but the converse is not necessarily true.
Nevertheless, when the minor is a subcubic graph, it is a topological minor too due to Diestel~\cite{Die17}.

\begin{proposition}{\rm\cite[Proposition~1.7.3]{Die17}}
\label{prop01}
If a subcubic graph $F$ is a minor of $H$, then $F$ is a topological minor of $H$ too.
\end{proposition}
Since every subdivision of the three-arm spider $T_k$ contains $T_k$ as a subgraph,
this proposition immediately implies that excluding $T_k$ as a subgraph in $H$ also excludes $T_k$ as a minor.
We summarize this conclusion in the next lemma.

\begin{lemma}
\label{lemma04}
If $H$ is a spider-free graph, then the three-arm spider $T_k$ is not a minor of $H$.
\end{lemma}

Given a graph $H$, a {\em tree decomposition} represents the vertices of $H$ as subtrees of a tree $T$,
in such a way that vertices in $H$ are adjacent only when the corresponding subtrees of $T$ overlap.
In detail, a tree decomposition~\cite[Section~12.3]{Die17} is a pair $(\mcB, T)$, where $\mcB = \{B_1, B_2, \dots, B_q\}$ is a family of subsets (called {\em bags}) of vertices of $H$,
and $T$ is a tree whose nodes are these bags, satisfying the following properties:
1) every vertex of $H$ appears in (called {\em associated with}) some bags that are connected in $T$, and
2) for every edge $(u, v)$ in $H$, both $u$ and $v$ appear in some bags.
The {\em width} of the tree decomposition is the size of its largest bag minus one, i.e., $\max_{1 \le i \le q} |B_i| - 1$,
and the {\em treewidth} $\tw(H)$ of the graph $H$ is the minimum width among all possible tree decompositions of $H$. 

A tree decomposition in which the underlying tree $T$ is a path graph is called a {\em path decomposition} of $H$,
and the width parameter derived from this path decomposition is called {\em pathwidth}.
Analogously, the {\em pathwidth} $\pw(H)$ of the graph $H$ is the minimum width among all possible path decompositions of $H$~\cite{Die17}.
We remark that for notational convenience, for a path decomposition $(\mcB, P)$, where $\mcB = \{B_1, B_2, \dots, B_q\}$ and $P$ is the underlying path,
we assume these bags are indexed to match their node order on the path $P$.
One sees that every graph has path and tree decompositions, for example a single bag of all the vertices of the graph.
Efforts are put in finding an {\em optimal} path or tree decomposition that achieves the minimum pathwidth or treewidth, correspondingly, followed by their applications.

We next use the following result on the pathwidth of a graph due to Bienstock et al.~\cite{BRS91} (where a forest is a graph consisting of one or multiple trees).

\begin{lemma}{\rm\cite{BRS91}}
\label{lemma05}
If a graph $H$ contains no forest $F$ as a minor, then its pathwidth $\pw(H) \le |V(F)| - 2$.
\end{lemma}

Since $T_k$ is a tree, Lemmas~\ref{lemma04} and~\ref{lemma05} give an explicit constant upper bound on the pathwidth of the spider-free graph $H$.

\begin{corollary}
\label{coro01}
For a spider-free graph $H$, its pathwidth
$\pw(H) \le p_k =
  \begin{cases}
    \dfrac {3k - 5}2, & k \mbox{ odd;}\\
    \dfrac {3k - 4}2, & k \mbox{ even.}
  \end{cases}$
\end{corollary}
\begin{proof}
One sees that when $k = 2m + 1$, $|V(T_k)| = 1 + 3m = \frac {3k - 1}2$, and thus $\pw(H) \le |V(T_k)| - 2 = \frac {3k - 5}2$.

When $k = 2m$, $|V(T_k)| = 1 + (m - 1) + m + m = 3m = \frac {3k}2$, and thus $\pw(H) \le |V(T_k)| - 2 = \frac {3k - 4}2$.
\end{proof}

For a fixed $k$, a bounded-width path decomposition of a spider-free graph $H$ can be constructed in linear time.
For example, one can first call Bodlaender's algorithm~\cite{Bod96} to construct a tree decomposition of width at most $p_k$ for $H$ in linear time,
and then call Bodlaender and Kloks's algorithm~\cite{BK96} to construct a path decomposition of width at most $p_k$ for $H$ in linear time.
In the sequel, we assume a path decomposition of width at most $p_k$ for $H$ is readily available.

\begin{definition}
\label{def03}
A {\em nice} path decomposition $(\mcB, P)$ for a graph $H$, where $\mcB = \{B_0, B_1, \dots, B_q \}$, is a path decomposition satisfying
1) the two end-bags are empty, i.e., $B_0 = B_q = \emptyset$, and
2) for each $i = 1, 2, \dots, q$,
either $B_i = B_{i-1} \cup \{v\}$ for some vertex $v$ of $H$ ({\em introduce} $v$) or $B_i = B_{i-1} \setminus \{v\}$ for some vertex $v$ of $H$({\em forget} $v$).
\end{definition}

The following standard transformation shows that imposing this nice form does not increase the width.

\begin{lemma}
\label{lemma06}
A path decomposition $(\mcB = \{B_1, B_2, \dots, B_q\}, P)$ of a graph $H$ can be transformed in linear time into a nice path decomposition
$(\mcB' = \{B'_0, B'_1, \dots, B'_{2n} \}, P')$ of the same width, where $n = |V(H)|$.
\end{lemma}
\begin{proof}
We construct the bags in the new path decomposition sequentially and explicitly, starting with the empty bag $B'_0 = \emptyset$.
Assuming without loss of generality that $B'_i$ has just been constructed using a sub-collection of bags $\{B_1, B_2, \ldots, B_{j-1}\}$ of $\mcB$,
and that $B'_i = B_{j-1}$, where $B_0 = \emptyset$ for notational convenience.
Forgetting one vertex of $B'_i \setminus B_j = B_{j-1} \setminus B_j$ at a time to produce $s = |B_{j-1} \setminus B_j|$ new bags sequentially,
namely, $B'_{i+1}, \dots, B'_{i+s}$.
One thus sees that $B'_{i+s} = B_{j-1} \cap B_j$.
Next, introducing one vertex of $B_j \setminus B_{j-1}$ at a time to produce $t = |B_j \setminus B_{j-1}|$ new bags sequentially,
namely, $B'_{i+s+1}, \dots, B'_{i+s+t}$.
One thus sees that $B'_{i+s+t} = B_j$.
This way, using the bag $B_j$, we have constructed $s + t = |B_{j-1} \setminus B_j| + |B_j \setminus B_{j-1}|$ sequential bags such that
these $i + s + t$ bags satisfy the second property in \Cref{def03}.

At the end of using the last bag $B_q$, we achieve a sequence of bags $B'_0, B'_1, \dots, B'_\ell$ such that $B'_\ell = B_q$.
Lastly, forget the vertices of $B_r$ one at a time from $B'_\ell$ to produce $|B_q|$ bags with the last one $B'_{q'} = \emptyset$, where $q' = \ell + |B_q|$.
Therefore, both properties in \Cref{def03} are satisfied.

For notational convenience, we denote $B_{q+1} = \emptyset$.
We have $q' = \sum_{j=1}^q |B_j \setminus B_{j-1}| + \sum_{j=2}^{q+1} |B_{j-1} \setminus B_j| = |V(H)| + |V(H)| = 2 |V(H)|$. i.e.,
every vertex of $H$ is introduced exactly once and forgotten exactly once.
This proves the number of constructed bags in $\mcB'$. 

It is seen that all the old bags of $\mcB$ occur in $\mcB'$, and thus every edge of $H$ continues to exist.
Also, by the ways of forgetting and introducing vertices from $B_{j-1}$ to achieve $B_{j-1} \cap B_j$ and then to achieve $B_j$,
every vertex of $H$ continues to associate with an {\em interval} of bags in $\mcB'$, and thus a subpath of the path $P'$.
That is, $\mcB'$ is a nice path decomposition for $H$.
Since we do not produce any bag of size larger than $\max_{1 \le j \le q} |B_j|$, the pathwidth of $\mcB'$ is the same as that of $\mcB$.
One clearly sees that the construction takes $O(n + q)$ time, which is linear.
This finishes the proof of the lemma.
\end{proof}

By \Cref{lemma06}, in the sequel, we assume a nice path decomposition $\mcB = \{B_0, B_1, \dots, B_{2n}\}$ of width at most $p_k$ for a spider-free graph $H$ is readily available,
where $n = |V(H)|$ and each vertex of $H$ is introduced first and then forgotten from a later bag.
We use $\mcB$ to solve the weighted $k$PVC problem on $H$ exactly.

\subsection{A dynamic programming exact algorithm}
Consider a spider-free graph $H$ and a nice path decomposition $\mcB = \{B_0, B_1, \dots, B_{2n}\}$ of width at most $p_k$ for $H$ stated right in the above.
Our algorithm examines these bags one by one sequentially, and distinguishes on how the current bag $B_i$ is obtained from $B_{i-1}$, by introducing or forgetting a vertex $v$.

Denote $U_i = \cup_{j=1}^i B_j$, the set of vertices introduced up to $B_i$, and the induced subgraph $H_i = H[U_i]$.
Importantly, since when a vertex is forgotten all its neighbors must have all been introduced,
there is no edge between the set $U_i \setminus B_i$ of vertices (which have been forgotten) and the set $V(H) \setminus U_{i-1}$ of vertices (which are to be introduced).
This leads to our idea to maintain the collection of all {\em partial} $k$-paths in the current graph $H[U_i \setminus C]$, where $C \subseteq U_i$ is the set of deleted vertices,
and update them when the vertex $v$ is introduced, deleted, or forgotten.
We point out that the existence of a {\em complete} $k$-path indicates the {\em infeasibility} of the deleted vertex set $C$.

Furthermore, we ignore the detailed vertices of $C \cap (U_i \setminus B_i)$, but only record their total weight.
On the other hand, we do keep the detailed vertices of $D = C \cap B_i$.
With respect to $D$, a partial $k$-path is denoted as a vector
\[
\sigma = (\sigma_1, \dots, \sigma_k) \in ((B_i \setminus D) \cup \{0, \star\})^k,
\]
where a vertex $v \in B_i \setminus D$ denotes a {\em retained} vertex, $0$ denotes an unused position,
and each occurrence of $\star$ denotes a distinct retained forgotten vertex of $(U_i \setminus B_i) \setminus C$.
The consecutive occupied positions correspond to adjacent distinct vertices.
Let $\Phi$ denote the collection of all such partial $k$-paths that are valid in the current residual graph;
define a dynamic programming (DP) state as a pair $(D, \Phi)$,
and ${\rm DP}_i[D, \Phi]$ stores the minimum weight of a deletion set in $U_i$ producing this state.
That is, if several deletion sets produce the same state, then only the minimum-weight one is kept for backtracking purpose.

At the initial empty bag $B_0 = \emptyset$, we initialize $\Phi = \{(0, \ldots, 0)\}$ and set ${\rm DP}_0[\emptyset, \Phi] = 0$.
Below we assume all possible states $(D, \Phi)$ at the bag $B_{i-1}$ have been decided, and their associated minimum weight ${\rm DP}_{i-1}[D, \Phi]$ computed.
We distinguish the two introduce and forget transitions.

\paragraph{Case 1: $B_i = B_{i-1} \cup \{v\}$.}
For each state $(D, \Phi)$ at $B_{i-1}$, the DP algorithm considers whether $v$ is deleted or retained.

If $v$ is deleted, then $v$ is added to $D$, its weight $w(v)$ is added to the objective value, and all the partial $k$-paths remain unchanged.
That is, a valid state $(D \cup \{v\}, \Phi)$ is reached and ${\rm DP}_i[D \cup \{v\}, \Phi] = {\rm DP}_{i-1}[D, \Phi] + w(v)$.

Suppose that $v$ is retained.
Firstly, each partial $k$-path $\sigma = (\sigma_1, \ldots, \sigma_k) \in \Phi$ remains valid, since $v$ does not have to be used in the path.
In addition, $v$ can be added the path $\sigma$ to produce new representations, by replacing each position $x$ with $\sigma_x = 0$ by $v$.
To be valid, every occupied position adjacent to $x$, i.e., $\sigma_{x-1}$ and $\sigma_{x+1}$ if they are within range,
has to have a vertex of $B_{i-1} \subset B_i$ adjacent to $v$ in the underlying graph $H$, or otherwise the new representation is invalid and deleted.
We note that an occupied neighboring position cannot be $\star$, which represents a forgotten vertex not adjacent to $v$ in the graph $H$.
The final collection, denoted as $\Phi'$, is the union of $\Phi$ and these new valid partial $k$-paths containing $v$.
We have ${\rm DP}_i[D, \Phi'] = {\rm DP}_{i-1}[D, \Phi]$.

\paragraph{Case 2: $B_i = B_{i-1} \setminus \{v\}$.}
For each state $(D, \Phi)$ at $B_{i-1}$, the DP algorithm considers whether $v$ has been deleted or is retained.

If $v \in D$, i.e., $v$ has been deleted, then it is removed from $D$ to update $D$ to be a subset of $B_i$,
and all the partial $k$-paths remain unchanged.
That is, a valid state $(D \setminus \{v\}, \Phi)$ is reached and ${\rm DP}_i[D \setminus \{v\}, \Phi] = {\rm DP}_{i-1}[D, \Phi]$.

Suppose that $v \notin D$, i.e., $v$ is retained.
For each partial $k$-path $\sigma = (\sigma_1, \ldots, \sigma_k) \in \Phi$, if it does not contain $v$, then it remains valid;
otherwise, the occurrence of $v$ is replaced by $\star$ while all other symbols remain unchanged.
We note that such a symbol-replacement process might lead to identical partial $k$-paths, and if so then only one is kept.
For the final collection, denoted as $\Phi'$, we have ${\rm DP}_i[D, \Phi'] = {\rm DP}_{i-1}[D, \Phi]$.

We also note that in the above two transitions, multiple states at $B_{i-1}$ can lead to the same state at $B_i$,
and if so then the algorithm keeps the minimum weight and the associated deletion set for backtracking purpose.
At the final empty bag $B_{2n}$, every remaining state corresponds exactly to a $k$PVC of the underlying graph $H$,
and the one achieving the minimum weight can be backtracked to produce the deletion set returned as the solution, denoted as ${\rm ALG}(H, w)$.

\begin{theorem}
\label{thm02}
For every fixed $k$, the dynamic programming algorithm is an $(n \cdot f(k))$-time exact algorithm for the weighted $k$PVC problem on spider-free graphs,
where $n$ is the order of the graph and $f(k)$ is a function depending only on $k$.
\end{theorem}
\begin{proof}
For a spider-free graph $H$, by Corollary~\ref{coro01}, $H$ has pathwidth at most $p_k < \frac 32 k$.
A nice path decomposition $\mcB$ of $H$ can be constructed in linear time.
Using $\mcB$, a minimum-weight $k$PVC $C^*$ of $H$ can walk through to delete or retain vertices correspondingly, reaching a valid state at the last bag $B_{2n}$ where $n = |V(H)|$.
It follows that ${\rm ALG}(H, w)$ has the same weight as $C^*$, and thus is an optimal solution.

Recall that in the DP algorithm we have states in the form $(D, \Phi)$,
where $D \subseteq B_i$ is the deletion set within the current bag $B_i$ and $\Phi$ is the collection of valid partial $k$-paths in the current residual graph.
Let
\[
b = \max_i |B_i| \le p_k + 1 < \frac 32 k.
\]
There are at most $(b + 2)^k$ possible partial $k$-paths.
Hence the number of possible collections $\Phi$ is at most $2^{(b + 2)^k}$, and there are at most $2^b$ choices for $D$.
Thus the number of states at each bag is no more than
\[
2^b \cdot 2^{(b + 2)^k} \le 2^{\frac 32 k + (\frac 32 k + 1)^k},
\]
which is a constant depending only on $k$.

By the two-case distinction in the DP algorithm, each transition examines a state at the preceding bag at most $k + 1$ times to produce all possible states at the current bag,
followed by duplicate removing.
Therefore each transition takes a time depending only on $k$.
We thus conclude that the overall running time for the DP algorithm, and consequently for the entire process leading to ${\rm ALG}(H, w)$,
is $n \cdot f(k)$, where $n = |V(H)|$ and $f(k)$ is a function depending only on $k$.
\end{proof}

\section{Algorithmic framework of layering and reverse extension}\label{sec03}
Given a subgraph $H$ of the current vertex positively weighted graph $G = (V, E)$,
a non-zero nonnegative function $c_H: V(H) \to \Qge$ that defines the vertex weights in $H$ is called a {\em local cost function} for $H$.
We use $(H, c_H)$ to denote the vertex non-negatively weighted graph $H$ in which $c_H$ is the vertex weight function.

\begin{definition}[Good local cost]
\label{def04}
Given a vertex subset $S \subseteq V$, a local cost function $c_H$ for $H = G[S]$, and a constant $\rho > 0$,
if $c_H(C \cap S) \le \rho \cdot \OPT(H, c_H)$ holds for every minimal $k$PVC $C$ of $G$, then we call the weighted graph $(H, c_H)$ {\em $\rho$-good} in $G$.
\end{definition}

Layering and reverse extension is an algorithmic framework that, given a weighted graph $G = (V, E)$,
cuts off a layer of vertex weights resulting in a {\em smaller} graph $G'$,
recursively computes a solution for $G'$,
then greedily extends the solution using some vertices in the layer to become a solution for the original graph $G$.
For an input graph $G = (V, E)$ in which the vertex weight function is $w(\cdot)$, we denote the procedure as {\sc LRE$(G, w)$},
in which the detailed steps of computation are depicted in \Cref{alg01}.

\begin{algorithm}[htpt]
\caption{{\sc LRE$(G, w)$}.\label{alg01}}
\begin{algorithmic}[1]
\Require
	$G = (V, E)$ with vertex weight function $w: V \rightarrow \Qge$;
\Ensure
	a minimal $k$PVC $C$.
\State $Z \leftarrow \{v \in V \mid w(v) = 0\}$;
\If {($Z \ne \emptyset$)}
	\State $V' \leftarrow V \setminus Z$;
	\State $G' \leftarrow G[V']$;
	\State $C \leftarrow$ {\sc LRE$(G', w|_{V \setminus Z})$};\Comment{recursive call;}
	\While {($Z \ne \emptyset$)}
		\State take a vertex $z$ of $Z$;\Comment{reverse extension starts;}
		\State $V' \leftarrow V' \cup \{z\}$;
		\State $G' \leftarrow G[V']$;
		\If {($C$ is not a $k$PVC of $G'$)}
			$C \leftarrow C \cup \{z\}$;
		\EndIf
		\State $Z \leftarrow Z \setminus \{z\}$;
	\EndWhile
	\State \textbf{return} \ $C$;
\EndIf
\If {($G$ is spider-free)}\Comment{boundary case;}
	\State \textbf{return} \ ${\rm ALG}(G, w)$;\Comment{a minimum-weight $k$PVC of $G$;}
\Else
	\State find a maximal-arm spider $H = \{c, Q_1, Q_2, \dots, Q_p\}$ in $G$;\Comment{layering starts;}
	\State define $c_H(v) = \begin{cases}
			p-1,	&v = c;\\
			1, 	&v \in V(H) \setminus \{c\};
			\end{cases}$
	\State $\lambda \leftarrow \max\left\{ \mu \ge 0 \mid w(v) - \mu \cdot c_H(v) \ge 0, \ \forall v \in V(H) \right\}$;
	\State define $\bar w(v) = \begin{cases}
			w(v) - \lambda \cdot c_H(v), 	&v \in V(H);\\
			w(v), 									&v \notin V(H);
			\end{cases}$
	\State \textbf{return} \ {\sc LRE}$(G, \bar w)$;\Comment{recursive call;}
\EndIf
\end{algorithmic}
\end{algorithm}

In Lines 13--20 of the layering phase, the conditional statement distinguishes whether a layering is necessary,
namely, if $G$ is spider-free, i.e., {\em boundary case},
then (Line 14) the dynamic programming algorithm stated in \Cref{thm02} is called to solve the instance exactly.
Otherwise, (Lines 16-20) a three-arm spider $T_k$ can be located and
then greedily a maximal-arm spider $H = \{c, Q_1, Q_2, \dots, Q_p\}$ containing $T_k$ is formed, where $p \ge 3$.
The vertices of this spider $H$ form a layer, for which a local cost function $c_H(\cdot)$ is designed (Line 17),
and subsequently the coefficient $\lambda$ of the layer is determined so that taking off the layer from the graph $G$ leaves at least one zero-weight vertex.
Then the algorithm recursively runs on the graph $G$ with the remaining vertex weights.
During the recursive call, in Lines 1--12, the algorithm first removes zero-weighted vertices,
then executes Lines 13--20 to find a solution $C$.
Upon $C$ is returned, the reverse extension phase starts, in which zero-weighted vertices are scanned in an arbitrary order to restore the graph,
so that if $C$ does not remain as a $k$PVC of the current graph, then the scanned vertex is added to $C$.
This ensures the final $C$ is a minimal $k$PVC for the original graph $G$.

\begin{lemma}[Reverse extension preserves minimality]
\label{lemma07}
Let $V' \subset V$ and $C$ be a minimal $k$PVC of the graph $G[V']$, and let $z \in V \setminus V'$ be a restored vertex.
If $C$ is a $k$PVC of $G[V' \cup \{z\}]$, then keep $C$; otherwise, replace it by $C^+ = C \cup \{z\}$.
In either case, the resulting set is a minimal $k$PVC of $G[V' \cup \{z\}]$.
\end{lemma}
\begin{proof}
The minimality holds trivially if $C$ is a $k$PVC of $G[V' \cup \{z\}]$.

Next suppose $C$ is not a $k$PVC of $G[V' \cup \{z\}]$.
One first sees that $C^+ = C \cup \{z\}$ is a $k$PVC of $G[V' \cup \{z\}]$.
If $C^+ = C \cup \{z\}$ is not minimal, then there exists an $x \in C$ such that $C \cup \{z\} \setminus \{x\}$ is a $k$PVC of $G[V' \cup \{z\}]$.
Since $z$ is not on any $k$-path of $G[V']$, $C \setminus \{x\}$ is a $k$PVC of $G[V']$, contradicting the minimality of $C$ in $G[V']$.
This proves that $C^+ = C \cup \{z\}$ is a minimal $k$PVC of $G[V' \cup \{z\}]$.
\end{proof}

\begin{lemma}[Spider local optimum]
\label{lemma08}
For the weighted $p$-arm spider $(H, c_H)$ defined in Lines 16--17 of \Cref{alg01}, $\OPT(H, c_H) \ge p - 1$.
\end{lemma}
\begin{proof}
Let $C$ be a $k$PVC of $H$.
If the center $c \in C$, then $c_H(C) \ge c_H(c) = p-1$.

Suppose that $c \notin C$.
We conclude from \Cref{lemma01} that $C$ contains at least $p - 1$ arm-vertices each of local cost $1$.
That is, $c_H(C) \ge p - 1$.
This proves the lemma.
\end{proof}

\begin{lemma}[Spider local upper bound]
\label{lemma09}
For the weighted maximal-arm spider $(H, c_H)$ defined in Lines 16--17 of \Cref{alg01},
every minimal $k$PVC $C$ of $G$ satisfies $c_H(C \cap V(H)) \le c_H(V(H)) - 1$.
\end{lemma}
\begin{proof}
Note that the vertices of $H$ have positive weights of at least $1$.
By \Cref{lemma02}, $C$ does not contain all the vertices of $H$, and thus $c_H(C \cap V(H)) \le c_H(V(H)) - 1$.
\end{proof}

\begin{lemma}[Good spider layer]
\label{lemma10}
The weighted maximal-arm spider $(H, c_H)$ defined in Lines 16--17 of \Cref{alg01} is $\rho$-good with
$\rho = \begin{cases}
	\frac {3k - 1}4,	&k \mbox{ is odd};\\
	\frac {3k}4,			&k \mbox{ is even}.
	\end{cases}$
\end{lemma}
\begin{proof}
Let $C$ be a minimal $k$PVC of $G$.
By Lemmas~\ref{lemma08} and~\ref{lemma09}, and $p \ge 3$, when $k$ is odd,
\[
\frac {c_H(C \cap V(H))}{\OPT(H, c_H)} \le \frac {pm + p - 2}{p - 1} = m + 1 + \frac {m - 1}{p - 1}
	\le m + 1 + \frac {m - 1}2 = \frac {3m + 1}2 = \frac {3k - 1}4.
\]
When $k$ is even,
\[
\frac {c_H(C \cap V(H))}{\OPT(H, c_H)} \le \frac {pm + p - 3}{p - 1} = m + 1 + \frac {m - 2}{p - 1}
	\le m + 1 + \frac {m - 2}2 = \frac {3m}2 = \frac {3k}4.
\]
That is, $(H, c_H)$ is $\rho$-good with the defined $\rho$.
\end{proof}

\begin{theorem}[Local-ratio theorem for {\sc LRE}]
\label{thm03}
Let $\rho \ge 1$.
If every local weighted maximal-arm spider $(H, c_H)$ defined in Lines 16--17 of \Cref{alg01} is $\rho$-good in the current graph
and ${\rm ALG}$ in Line 14 solves spider-free graphs exactly,
then {\sc LRE}$(G, w)$ returns a minimal $k$PVC $C$ of $G$ with weight $w(C) \le \rho \cdot \OPT(G, w)$.
\end{theorem}
\begin{proof}
Denote $(G_0, w_0) = (G, w)$ the original graph,
and $G_i$ denotes the subgraph induced by the vertices having positive residual weight after the $i$-th layering step,
and let $w_i$ be the corresponding residual weight function restricted to $V(G_i)$ (that is, for every $v \notin V(G_i)$, $w_i(v) = 0$).
The local weighted maximal-arm spider $(H_i, c_{H_i})$ of the $i$-th layer is selected from the graph $G_{i-1}$ with coefficient $\lambda_i$,
for $i = 1, 2, \dots, t$, and terminates at the residual spider-free graph $(G_t, w_t)$.

Let $C$ be the $k$PVC returned by {\sc LRE}, and let $C^*$ be an optimum $k$PVC of $(G, w)$.
For every layer $i$, let $\OPT(H_i, c_{H_i})$ denote the minimum-weight of the $k$PVC of $H_i$.
Since $C^* \cap V(H_i)$ is a $k$PVC of $H_i$ and $c_{H_i}$ is restricted to $V(H_i)$, we have
\[
c_{H_i}(C^*) = c_{H_i}(C^* \cap V(H_i)) \ge \OPT(H_i, c_{H_i}), \ i = 1, 2, \dots, t.
\]
Similarly, $C^* \cap V(G_t)$ is a $k$PVC of the spider-free graph $G_t$, and hence
\[
w_t(C^*) = w_t(C^* \cap V(G_t)) \ge \OPT(G_t, w_t).
\]
Therefore,
\[
\OPT(G, w) = w(C^*) = w_t(C^*) + \sum_{i = 1}^t \lambda_i \cdot c_{H_i}(C^*) \ge \OPT(G_t, w_t) + \sum_{i = 1}^t \lambda_i \cdot \OPT(H_i, c_{H_i}).
\]

Recall that the returned $C$ is formed by ${\rm ALG}(G_t, w_t)$, followed by reverse extension that adds vertices from $V(G_{i-1}) \setminus V(G_i) \subseteq V(H_i)$ into $C$,
sequentially for $i = t, t-1, \dots, 1$.
For every layer $i$, consider the moment in the reverse-extension process when all vertices of $G_{i-1}$ have been restored.
By Lemma~\ref{lemma07}, the solution at this moment is a minimal $k$PVC of $G_{i-1}$.
Moreover, subsequent reverse-extension steps never remove any vertex from the solution.
Therefore, the solution at this moment is exactly $C \cap V(G_{i-1})$.
That is, $C \cap V(G_{i-1})$ is a minimal $k$PVC of $G_{i-1}$.
Since $(H_i, c_{H_i})$ is $\rho$-good in $G_{i-1}$,
\[
c_{H_i}(C) = c_{H_i}(C \cap V(H_i)) \le \rho \cdot \OPT(H_i, c_{H_i}).
\]

Using $w_t(C) = w_t({\rm ALG}(G_t, w_t)) = \OPT(G_t, w_t)$ and $\rho \ge 1$, we have
\[
w(C) = w_t(C) + \sum_{i = 1}^t \lambda_i \cdot c_{H_i}(C)
	\le \OPT(G_t, w_t) + \rho \cdot \sum_{i = 1}^t \lambda_i \OPT(H_i, c_{H_i})
	\le \rho \cdot \OPT(G, w).
\]
This proves the theorem.
\end{proof}

\begin{lemma}
\label{lemma11}
Given a vertex non-negatively weighted graph $G = (V, E)$, the {\sc LRE} algorithm runs in $O(n^{h(k)} \cdot {\rm poly}(k) + n \cdot f(k))$ time,
where $n = |V|$ and $h(k)$, ${\rm poly}(k)$ and $f(k)$ are functions depending only on $k$.
\end{lemma}
\begin{proof}
A three-arm spider $T_k$ contains at most $\frac 32 k$ vertices interconnected in a specific structure,
and thus it can be located, if exists, in $O(n^{\frac 32 k} \cdot {\rm poly}(k))$ time.
Afterwards, since enumerating all possible arms at the center vertex in the subgraph induced by the vertices outside of the spider takes
$O(\frac k2 \cdot n^{\frac k2})$ time, and there are $O(n^{\frac k2})$ of them,
greedily expanding $T_k$ into a maximal-arm spider $H$, in which there are at most $\frac {2n}k$ arms,
takes $O(\frac k2 \cdot n^{\frac k2} \cdot {\frac {2n}k}) = O(n^{\frac k2 + 1})$ time.

Given a maximal-arm spider $H$, setting up the local cost function $c_H$, determining the parameter $\lambda$,
and updating the vertex weight function $\bar{w}$ takes $O(n + k)$ time.
It follows that every layer takes an overall time of $O(n^{\frac 32 k} \cdot {\rm poly}(k))$,
and it leaves at least a zero-weight vertex in the residual graph.
That is, there are at most $n$ layers.

During the reverse extension, when a zero-weight vertex is restored, the feasibility of the current deletion set $C$ can be tested
by enumerating all $k$-paths in the restored graph, which takes $O(n^k \cdot {\rm poly}(k))$ time.
Since there are at most $n$ restorations, the overall reverse extension time is in $O(n^{k+1} \cdot {\rm poly}(k))$.

Recall from \Cref{thm02} that the dynamic programming algorithm produces an exact solution for a spider-free graph, which is a subgraph of $G = (V, E)$,
in $O(n \cdot f(k))$ time, where $f(k)$ is a function depending only on $k$.
Summing up, a coarse running time bound for the LRE is $O(n^{\frac 32 k + 1} \cdot {\rm poly}(k) + n \cdot f(k))$,
where ${\rm poly}(k)$ is a polynomial function in $k$ and $f(k)$ is an exponential function depending only on $k$.
We remark that if $k$ is part of the input, then the running time is not necessarily polynomial any more.
\end{proof}

Combining \Cref{lemma10}, \Cref{thm03}, and \Cref{lemma11}, we have achieved the conclusion stated as \Cref{thm01}, that is,
for every fixed integer $k \ge 3$, the weighted $k$PVC problem admits a polynomial-time $\rho_k$-approximation algorithm, where
\[
\rho_k = \begin{cases}
	\frac {3k - 1}4, & k \mbox{ is odd};\\
	\frac {3k}4,     & k \mbox{ is even}.
	\end{cases}
\]

\section{Conclusion}\label{sec04}
We studied the weighted $k$-Path Vertex Cover problem from the perspective of approximation algorithm, where the input graph is vertex non-negatively weighted.
We observed a maximal-arm spider that can be approximated within a factor of $\frac {3k}4$, and thus designed it into layers in the Layering and Reverse Extension framework.
Furthermore, we proposed a polynomial-time exact dynamic programming algorithm that solves the boundary case of spider-free graphs,
based on a path decomposition of bounded width.
Using the local-ratio technique, we were able to prove that the performance ratio of our approximation algorithm is no more than $\frac {3k}4$.
The improvement in the approximation ratio is huge considering the previous best ratio being $k - \frac 12$, which was improved from an earlier ratio of $k$.
Notably, the running time of both the dynamic programming algorithm, and the Layering and Reverse Extension algorithm, is exponential in $k$, and thus requires $k$ to be a fixed constant.
It would be interesting to see whether the approximation ratio lower bound is $O(k)$ or not,
and if not, how one can improve the ratio into $o(k)$.


\end{document}